\documentclass[runningheads]{llncs}

\usepackage[T1]{fontenc}
\usepackage{hyperref}
\usepackage{url}
\usepackage{graphicx}
\usepackage{wrapfig}
\usepackage{multirow}

\usepackage{bbm}
\usepackage{amssymb}
\usepackage{amsthm}
\usepackage{xcolor}
\usepackage{amsmath}

\newtheorem{assumption}{Assumption}

\def\R{{\mathbb{R}}}

\def\N{{\mathbb{N}}}

\newcommand{\norm}[1]{\left\lVert#1\right\rVert}
\newcommand{\dotprod}[1]{\left< #1\right>}

\newcommand{\eE}{\mathbb{E}}
\DeclareMathOperator*{\argmin}{arg\,min}
\newcommand{\bpar}[1]{\left(#1\right)}

\author{Marien Renaud\thanks{corresponding author : marien.renaud@math.u-bordeaux.fr}\inst{1} \and
Julien Hermant\inst{1} \and
Nicolas Papadakis\inst{1}}

\institute{Univ. Bordeaux, CNRS, INRIA, Bordeaux INP, IMB, UMR 5251, F-33400 Talence, France, Corresponding author : marien.renaud@math.u-bordeaux.fr}

\begin{document}
\title{Convergence Analysis of a Proximal Stochastic Denoising Regularization Algorithm
}
\titlerunning={Analysis of a Proximal Stochastic Denoising Regularization Algorithm}
\date{}
\maketitle

\begin{abstract}
Plug-and-Play methods for image restoration are iterative algorithms that solve a variational problem to recover a clean image from a degraded observation. 
These algorithms are known to be flexible to changes of degradation and to perform state-of-the-art restoration.
Recently, significant efforts have been made to explore new stochastic algorithms based on  the Plug-and-Play or REgularization by Denoising (RED) frameworks, such as SNORE, which is a convergent stochastic gradient descent algorithm.
A variant of this algorithm, named SNORE Prox,  reaches state-of-the-art performances, especially for inpainting tasks. 
However, the convergence of SNORE Prox, that can be seen as a stochastic proximal gradient descent, has not been analyzed so far. 
In this paper, we prove the convergence of SNORE Prox under non convex assumptions.

\keywords{Image restoration, stochastic optimization, plug-and-play}
\end{abstract}

\section{Introduction}

\textbf{Image restoration} Recovering a clean image $x \in \R^d$ from a degraded observation $y \in \R^m$ is called image restoration. When a model of physical degradation is known, we have $y \sim \mathcal{N}(\mathcal{A}(x))$, with $\mathcal{A}:\R^d\to \R^m$ a deterministic operator and $\mathcal{N}$ a noise distribution. For example, a linear degradation with additive Gaussian noise can be written $y = A x + n$, with $A \in \R^{m\times d}$ and $n \sim \mathcal{N}(0, \sigma_y^2 I_m)$. With a Bayesian interpretation, image restoration is equivalent to solving the following variational problem
\begin{align}\label{eq:ideal_problem}
    \argmin_{x \in \R^d} F(x) := f(x) + \lambda g(x),
\end{align}
where $f = - \log p(\cdot|y)$ is the data-fidelity term that depends on the model of the degradation, $g = - \log p$ is the regularization term that encodes the prior model set on clean images and $\lambda$ is a regularization weight. With this formulation, the knowledge of the degradation and clean images is divided in two terms.

\textbf{The choice of learned regularization prior  is crucial.} Defining a relevant regularization is fundamental to achieve efficient restoration~\cite{mallat1999wavelet,rudin1992nonlinear}. In particular, learned regularizations such as deep image priors~\cite{chen2024bagged,ulyanov2018deep}, Plug-and-Play (PnP)~\cite{venkatakrishnan2013plug} or Regularization by Denoising (RED)~\cite{romano2017little}, provide state-of-the-art restorations.

PnP and RED define a model on clean images by applying a Gaussian denoiser $D_{\sigma}$ at each iteration of the optimization path:
\begin{align}
    x_{k+1} &= D_{\sigma} \left( x_k - \delta \nabla f(x_k) \right) \tag{PnP} \label{eq:PnP} \\
    x_{k+1} &= x_k - \delta \nabla f(x_k) - \frac{\delta \lambda}{\sigma^2}\left(x_k - D_{\sigma}(x_k)\right). \tag{RED} \label{eq:RED}
\end{align}
In \eqref{eq:PnP}, applying the denoiser $D_{\sigma}$ is interpreted as a proximal step on the regularization~$g$. In \eqref{eq:RED}, thanks to the Tweedie formula~\cite{efron2011tweedie}, the residual of the denoiser $\frac{1}{\sigma^2}\left(x - D_{\sigma}(x)\right)$ is interpreted as a gradient step on the regularization~$g$.  The Gaussian denoiser is typically a learned deep neural network~\cite{zhang2021plug}.

Notice that iterates convergence guarantees can be derived for PnP or RED methods when constraining the denoiser~\cite{hurault2022gradient,hurault2023convergentplugandplayproximaldenoiser,pesquet2021learning,ryu2019plug,sun2019block,Sun_2019,sun2021scalable,wei2024learning}. 
In this work, we will leverage on the gradient-step denoiser~\cite{hurault2022gradient,hurault2023convergentplugandplayproximaldenoiser} that writes $D_\sigma(x)=x-\nabla h_\sigma (x)$ with~$h_\sigma$ a non-convex regularization.
This structural assumption on the denoiser ensures that PnP and RED schemes optimize an explicit regularization function.

\textbf{Improved regularization with stochastic PnP or RED.} Recent works have focus on developing new PnP algorithms to improve the restoration quality~\cite{hu2023restoration,shumaylov2024weakly}.
Among these works, stochastic versions of PnP appear as a promising line of work to lower the computational cost~\cite{sun2019block,tang2020fast}. It also allows better restorations to be achieved~\cite{hu2024stochastic,laumont2023maximum}. In particular, the Stochastic deNOising REgularization (SNORE) proposed in~\cite{renaud2024plugandplayimagerestorationstochastic} is defined by
\begin{equation}
    x_{k+1} = x_k - \delta \nabla f(x_k) - \frac{\delta \lambda}{\sigma^2}\left(x_k - D_{\sigma}(x_k + \sigma z_{k+1})\right), \tag{SNORE} \label{eq:SNORE}
\end{equation}
with $z_{k+1} \sim \mathcal{N}(0, I_d)$. \eqref{eq:SNORE} ensures that the denoiser of level $\sigma$ is applied to images that contain the corresponding level of Gaussian noise.
As a side result, in~\cite{renaud2024plugandplayimagerestorationstochastic} is introduced SNORE Prox~\cite{renaud2024plugandplayimagerestorationstochastic} defined by
\begin{align}
    x_{k+1} = \text{Prox}_{\delta f}\left( x_k - \frac{\delta \lambda}{\sigma^2}\left(x_k - D_{\sigma}(x_k + \sigma z_{k+1}) \right) \right), \tag{SNORE Prox} \label{eq:SNORE_Prox}
\end{align}
where $\text{Prox}_{\delta f}(x):=\argmin_{z \in \R^d} \frac{1}{2\delta}\|x-z\|^2+f(z)$.
SNORE Prox shows state-of-the-art restoration results, especially for inpainting, but no convergence guarantees are provided for this algorithm.
Such guarantees are nevertheless crucial to ensure that the algorithm truly solves problem~\eqref{eq:ideal_problem}. As a by-product, convergence analysis can also help choosing hyper-parameters, such as the step-size~\cite{hurault2022gradient}.

\textbf{SNORE Prox is a Stochastic Proximal Gradient Descent (SPGD).}
As detailed in the next section (see relation~\eqref{eq:snoreprox}), SNORE Prox can be reformulated as a SPGD algorithm applied to problem~\eqref{eq:ideal_problem}. For image applications, it involves a non-convex regularization $g$ and a data-fidelity $f$ that might be weakly convex. 
SPGD has been introduced in~\cite{nitanda2014stochastic,xiao2014proximal} and then extensively studied in the literature. However, as illustrated in Table~\ref{tab:existing_cvg_results}, there exists no theoretical framework that exactly covers our case of interest, as convergence proofs only exist for $f$ being either convex~\cite{nitanda2014stochastic,xiao2014proximal,j2016proximal,ghadimi2016mini,atchade2017perturbed,allen2018katyusha,xu2023momentum,ding2023nonconvex} or weakly convex and Lipschitz~\cite{li2022unified}.

\begin{table}[ht]
    \centering
    \begin{tabular}{|c|c|c|c|}
        \hline
         References & Conditions on $f$ & Conditions on $g$\\
         \hline
        \cite{allen2018katyusha,atchade2017perturbed,nitanda2014stochastic} & convex & convex, $L$-smooth \\
        \cite{ding2023nonconvex,ghadimi2016mini,milzarek2023convergence,xu2023momentum} & convex & $L$-smooth, bounded variance \\
        \cite{j2016proximal,xiao2014proximal} & convex & $L$-smooth, finite sum \\
        \cite{li2022unified} &  weakly convex, Lipschitz & $L$-smooth, bounded variance \\
        this paper & weakly convex, $M$-smooth & $L$-smooth, bounded variance \\
        \hline
    \end{tabular}
    \vspace{0.2cm}
    \caption{Existing convergence results in the literature for SPGD. 
    }
    \label{tab:existing_cvg_results}\vspace*{-0.7cm} 
\end{table}

Moreover, most of the results in the literature are expressed in terms of proximal maps~\cite{ghadimi2016mini,xu2023momentum} that are related to the residual $\|x_{k+1} -x_k \|$. In this paper, we focus on convergence results in term of $\|\nabla F(x_k)\|$, which is more directly related to the convergence of the iterates to the critical points of $F$. 

\textbf{Contributions}
\textbf{(a)} In order to target a large class of data-fidelity terms in problem~\eqref{eq:ideal_problem}, we first generalize in Section~\ref{sec:residual} the SPGD residual convergence result of~\cite[Theorem 2]{ghadimi2016mini}, in the case of non-convex regularization and weakly convex data fidelity (Lemma~\ref{lemma:iterates_control_snore_prox_weakly}).
\textbf{(b)} In section~\ref{sec:criticalpoint}, we obtain critical point convergence guarantees, in terms of $\|\nabla F(x_k)\|$, on the iterations of SNORE Prox~\cite{renaud2024plugandplayimagerestorationstochastic} in both the constant step-size regime (Proposition~\ref{prop:constant_step-size_convergence}) and the non-increasing step-size one (Propositions~\ref{prop:varying_step-sizes} and~\ref{prop:decreasing_to_zero}).

\section{Convergence analysis: residuals}\label{sec:residual}
In this section, we study the convergence of residuals $||x_{k+1}-x_k||$ associated to the iterates of the SNORE Prox algorithm~\cite{renaud2024plugandplayimagerestorationstochastic} applied to the problem~\eqref{eq:ideal_problem} for a weakly convex function $f$ and a non-convex function $g$.  

For later purpose, we reformulate the SNORE Prox algorithm such that it includes both constant and non-increasing step-size regimes, as
\begin{equation}\label{snore_prox}
    x_{k+1} = \text{Prox}_{\delta_k f}\left( x_k - \frac{\delta_k \lambda}{\sigma^2}\left(x_k - D_{\sigma}(x_k + \sigma z_{k+1}) \right) \right).
\end{equation}
To make SNORE-Prox a Stochastic Proximal Gradient Descent (SPGD) algorithm, we introduce the following assumption on $D_{\sigma}$.
\begin{assumption}\label{ass:den_gradient_step}
    There exists a potential $h_{\sigma} : \R^d \to \R_+$ for the denoiser $D_{\sigma}$, i.e. $D_{\sigma} = I_d - \nabla h_{\sigma}$.
\end{assumption}
Assumption~\ref{ass:den_gradient_step} is verified for the exact MMSE denoiser, due to the Tweedie formula~\cite{efron2011tweedie}, with $h_{\sigma} = -\sigma^2 \log p_{\sigma}$, where $p_{\sigma}$ is the convolution $p \star \mathcal{N}_{\sigma}$ between the prior $p$ and the Gaussian kernel $\mathcal{N}_{\sigma} = \mathcal{N}(0,\sigma^2 I_d)$. Assumption~\ref{ass:den_gradient_step} is also verified by construction when considering a gradient-step denoiser~\cite{hurault2022gradient}, which will be used in our experiments.

Under Assumption~\ref{ass:den_gradient_step}, SNORE Prox is a stochastic proximal gradient descent applied to problem~\eqref{eq:ideal_problem} with the regularization
\begin{align}
    g(x) = g_{\sigma}(x) = \frac{1}{\sigma^2} \eE_{z \sim \mathcal{N}(0, I_d)}\left( h_{\sigma}(x + \sigma z) \right) = - \mathcal{N}_\sigma \star \log(p \star \mathcal{N}_\sigma).
\end{align}
Due to the form of $D_{\sigma}$ (Assumption~\ref{ass:den_gradient_step}), we get \begin{align}\label{eq:g}
\nabla g_{\sigma}(x) = \frac{1}{\sigma^2}\left( x -  \eE_{z \sim \mathcal{N}(0, I_d)}\left(D_{\sigma}(x + \sigma z)\right)\right).\end{align} In practice, we can only compute a stochastic approximation of $\nabla g_{\sigma}$ as
\begin{align}
\tilde \nabla g_{\sigma}(x) = \frac{1}{\sigma^2} \left(x - D_{\sigma}(x + \sigma z)\right),~z \sim \mathcal{N}(0, I_d),
\end{align}  
with the corresponding bias $\zeta_k = \tilde \nabla g_{\sigma}(x_k) - \nabla g_{\sigma}(x_k)$. 

Thus the SNORE Prox algorithm~\eqref{snore_prox} can be re-written in the form of a SPGD algorithm as
\begin{equation}
    x_{k+1} = \text{Prox}_{\delta_k f}\left( x_k - \delta_k \lambda \tilde \nabla g_{\sigma}(x_k) \right).\label{eq:snoreprox}
\end{equation}
To study the convergence of this scheme we make the following assumptions.
\begin{assumption}\label{ass:regularities_assumptions}
    (a) The denoiser $D_{\sigma}$ is $L$-Lipschitz, with $L \in \R_+$.

    \noindent
    (b) $f$ is differentiable and $\rho$-weakly convex, i.e. $f + \frac{\rho}{2}\|\cdot\|^2$ is convex, with $\rho \in \R_+$.

        \noindent
    (c) $F$ admit a lower bound $F^\ast \in \R$, i.e. $\forall x \in \R^d, F(x) \ge F^{\ast}$.
\end{assumption}
Assumption~\ref{ass:regularities_assumptions}(a) is verified for a neural network denoiser with Lipschitz activation function such as eLU. 
Note that no constraint on $L \in \R_+$ is required in the training of the denoiser to obtain convergence guarantees, contrary to~\cite{hurault2023convergentplugandplayproximaldenoiser}. However, having a small constant $L$ is crucial to use large step-sizes in the algorithm.
Assumption~\ref{ass:regularities_assumptions}(b) is verified for linear degradation, including inpainting, deblurring or super-resolution, with additive Gaussian noise. Assumption~\ref{ass:regularities_assumptions}(c) is necessary to ensure that problem~\eqref{eq:ideal_problem} is well defined and verified in practice.

Before stating the convergence of the residuals, we give a regularity result for $g_\sigma$.
\begin{lemma}\label{prop:gsmooth}
Under Assumptions~\ref{ass:den_gradient_step} and~\ref{ass:regularities_assumptions}(a), $\nabla g_\sigma$ is $\frac{L+1}{\sigma^2}$-Lipschitz.
\end{lemma}
\begin{proof}
Assumptions~\ref{ass:den_gradient_step}, relation~\eqref{eq:g} and Assumption~\ref{ass:regularities_assumptions}(a), give the result.
\end{proof}

\begin{lemma}\label{lemma:iterates_control_snore_prox_weakly}
Under Assumptions~\ref{ass:den_gradient_step} and~\ref{ass:regularities_assumptions}, for $\delta_0 \le \frac{\sigma^2}{\lambda (L+1) + \rho \sigma^2}$ and $(\delta_k)_{k \in \N}$ a non-increasing sequence of step-size, we get
\begin{align}
    \sum_{k=0}^{N-1} \eE_k\left(\|x_{k+1} - x_k\|^2\right) \le 2\delta_0 (F(x_0) - F^\ast) + \frac{4 \lambda^2 L^2}{\sigma^2 (1 - \delta_0 \rho)} \sum_{k=0}^N \delta_k^2.
\end{align}
\end{lemma}
Lemma~\ref{lemma:iterates_control_snore_prox_weakly} generalizes~\cite[Theorem 2]{ghadimi2016mini} to include weakly convex functions $f$. The proof of this result is postponed in Appendix~\ref{sec:proof_lemma_weakly}.  Lemma~\ref{lemma:iterates_control_snore_prox_weakly} can be reformulated as a control of the proximal map $G_k = \frac{x_k - x_{k+1}}{\delta_k} = \nabla f(x_{k+1}) + \lambda \tilde \nabla g(x_k)$, which is an implicit-explicit first order derivative of $F$. We formulate this result as a control on the residual $\|x_{k+1} - x_k\|$ to make it more intuitive.

\begin{assumption}\label{ass:step-size_to_zero}
    The step-size $(\delta_k)_{k \in \N}$ are non-increasing and $\sum_{k=0}^{+\infty} \delta_k^2 < +\infty$.
\end{assumption}

\begin{corollary}\label{cor:residual_to_zero}
Under Assumptions~\ref{ass:den_gradient_step},~\ref{ass:regularities_assumptions} and~\ref{ass:step-size_to_zero}, the residual of SNORE Prox iterates $\|x_{k+1} - x_k\|$ converges to $0$ in expectation.
\end{corollary}

Corollary~\ref{cor:residual_to_zero} ensures that residuals go to zero, even if the step-sizes decreases slowly, for instance $\sum_{k=0}^{+\infty} \delta_k = +\infty$.

\section{Convergence analysis: critical points}\label{sec:criticalpoint}
In this section, we show convergence results of the SNORE Prox iterates~\eqref{snore_prox} to a critical point $x^*$ of the function $F$, characterized by $\nabla F(x^*)=0$.  In subsection~\ref{ssec:main_result}, we present our main technical convergence result on $\|\nabla F(x_k)\|$, which requires a mild smoothness assumption on the data-fidelity $f$. We then deduce critical point convergence of the SNORE Prox algorithm with constant step-sizes (subsection~\ref{ssec:constant_step}) and non-increasing step-sizes (subsection~\ref{ssec:varying_step}).

\subsection{Main technical result}\label{ssec:main_result}
We now present our main technical result, Lemma~\ref{lemma:intermediate_result}, about critical point convergence analysis of the SNORE Prox iterates. This Lemma  provides information on the behavior of $\|\nabla F(x_k)\|$, encompassing  both constant and non-increasing step-sizes $( \delta_k )_{k \in \N}$. 
For that purpose, we first introduce a new smoothness assumption on the data-fidelity term $f$.
\begin{assumption}\label{ass:f_smooth}
    $f$ is $M$-smooth, \textit{i.e.}
    $\forall x,y \in \R^d,~ f(x) \leq f(y) + \langle \nabla f(y),x-y \rangle + \frac{M}{2}\lVert x-y \rVert^2.$
\end{assumption}
Assumption~\ref{ass:f_smooth} gives a control on the upper curvature of $f$.

\begin{remark}\label{remark:on_f_M_smooth}
It is not usual in the literature to impose smoothness on both functions
$f$ and $g_{\sigma}$~\cite{xu2023momentum,xiao2014proximal,allen2018katyusha,nitanda2014stochastic,j2016proximal,ghadimi2016mini,ding2023nonconvex}. However, in our setting, Assumption~\ref{ass:f_smooth} is necessary to obtain a convergence guarantee of the form $\min_{k=0,\dots,N}\eE[\norm{\nabla F(x_k)}^2]$ goes to zero with $N$.
Indeed, the authors of~\cite[Proposition 1]{gao2024non} highlight the following counter-example. For $(x_k)_{k \in \N}$ the iterates of SNORE Prox, taking $F(x) = f(x) + g_{\sigma}(x)$, with $f = \frac{a}{2}\norm{\cdot}^2$, $g_{\sigma} = \frac{\lambda}{2}\norm{\cdot}^2$ and $a \ge \frac{1}{\delta_k}$ implies that $\forall N \in \N,~ \min_{k=0,\dots,N}\eE[\norm{\nabla F(x_k)}^2] \ge \frac{\sigma^2}{4}$.
\end{remark}

\begin{lemma}\label{lemma:intermediate_result} 
Under Assumptions~\ref{ass:den_gradient_step},~\ref{ass:regularities_assumptions} and~\ref{ass:f_smooth}, with a non-increasing sequence of step-size verifying $\delta_0 \le \frac{\sigma^2}{\lambda (L+1) + \rho \sigma^2}$, there exist $A_1, B_1 \in \R_+$ such that the iterates of the SNORE Prox algorithm verify
\begin{align}\label{eq:bound_sum_gradients}
      \sum_{k=0}^N  \frac{\delta_k}{2}\eE\left[\norm{\nabla F(x_k)}^2 \right]
      \le A_1 \left( F(x_0) - F^\ast \right) + B_1  \sum_{k=0}^N \delta_k^2.
\end{align}
\end{lemma}
\begin{proof}
Using the fact that $g_\sigma$ is $\frac{L+1}{\sigma^2}$-smooth from Lemma~\ref{prop:gsmooth} together with Assumption~\ref{ass:f_smooth}, we get that $F=f+\lambda g_\sigma$ is $L_F=M + \frac{\lambda(L+1)}{\sigma^2}$-smooth:
\begin{equation*}
    F(x_{k+1}) \leq F(x_k) + \dotprod{\nabla F(x_k),x_{k+1}-x_k} + \frac{L_F}{2}\norm{x_{k+1}-x_k}^2.
\end{equation*}
Taking conditional expectation $\mathbb{E}_k \left[ \cdot \right] = \mathbb{E} \left[ \cdot | x_k \right]$, we obtain
    \begin{align}
   \hspace{-5pt}\mathbb{E}_k \left[ F(x_{k+1}) \right] &\leq \mathbb{E}_k \left[  F(x_k)\hspace{-1pt} + \hspace{-1pt}\dotprod{\nabla F(x_k),x_{k+1}-x_k}\hspace{-1pt} +\hspace{-1pt} \frac{L_F}{2}\norm{x_{k+1}-x_k}^2 \right]\nonumber\\
    &= F(x_k)\hspace{-1pt} + \hspace{-1pt}\dotprod{\nabla F(x_k),\mathbb{E}_k\left[ x_{k+1}-x_k \right]} \hspace{-1pt}+\hspace{-1pt} \frac{L_F}{2}\mathbb{E}_k\left[ \norm{x_{k+1}-x_k}^2 \right]\hspace{-1pt}.\label{eq:l_smooth_develop}
\end{align}
The optimal condition of the proximal operator in~\eqref{eq:snoreprox} implies that 
\begin{align*}
    x_{k+1}-x_k = -\delta_k\bpar{ \lambda \tilde \nabla g_{\sigma}(x_{k}) + \nabla f(x_{k+1})}.
\end{align*}
Taking conditional expectation, we obtain
\begin{align}
    \mathbb{E}_k\left[ x_{k+1}-x_k \right] = -\delta_k \bpar{ \lambda \nabla g_{\sigma}(x_{k}) +\mathbb{E}_k\left[  \nabla f(x_{k+1}) \right]}.\label{eq:optimal_condition_exp}
\end{align}
Combining equations~(\ref{eq:optimal_condition_exp}) and~(\ref{eq:l_smooth_develop}) and recalling that $\nabla F=\nabla f+\lambda \nabla g_\sigma$, we get
\begin{align}
    &\,\mathbb{E}_k \left[ F(x_{k+1}) \right]\nonumber \\\le&\, F(x_k) -\delta_k\dotprod{\nabla F(x_k),\lambda \nabla g_{\sigma}(x_{k}) +\mathbb{E}_k\left[  \nabla f(x_{k+1}) \right]} + \frac{L_F}{2}\mathbb{E}_k\left[ \norm{x_{k+1}-x_k}^2 \right] \nonumber\\
    =&\,F(x_k) -\delta_k \norm{ \nabla F(x_k)}^2  + \frac{L_F}{2}\mathbb{E}_k\left[ \norm{x_{k+1}-x_k}^2 \right] \nonumber\\&+ \delta_k\dotprod{\nabla F(x_k), \nabla f(x_k) -\mathbb{E}_k \left[  \nabla f(x_{k+1}) \right]} \nonumber\\
    \leq &\,F(x_k) - \frac{\delta_k}{2} \norm{ \nabla F(x_k)}^2  + \frac{L_F}{2}\mathbb{E}_k\left[ \norm{x_{k+1}-x_k}^2 \right] \nonumber\\&+ \frac{\delta_k}{2}\norm{\nabla f(x_k) -\mathbb{E}_k \left[  \nabla f(x_{k+1}) \right]}^2, \label{eq:l_smooth_develop_II}
\end{align}
where the last inequality uses $\dotprod{a,b} \leq \frac{1}{2}\norm{a}^2 + \frac{1}{2}\norm{b}^2$.

As $f$ verifies both Assumptions~\ref{ass:regularities_assumptions}(b) and \ref{ass:f_smooth}, its gradient is $\overline{M} := \max (\rho,M)$-Lipschitz. Since $\norm{\mathbb{E}\left[ \nabla f(x_k) - \nabla f(x_{k+1}) \right]}^2 \leq \mathbb{E}\left[\norm{ \nabla f(x_k) - \nabla f(x_{k+1})}^2 \right]$, we use this $\overline{M}$-Lipschitz property on $\nabla f$ and rearrange terms to get
\begin{align}\label{eq:bound_gradient}
    \frac{\delta_k}{2}\norm{\nabla F(x_k)}^2 &\leq F(x_k) -  \mathbb{E}_k \left[ F(x_{k+1}) \right] + \frac{L_F + \delta_k \overline{M}^2}{2} \mathbb{E}_k\hspace{-0.05cm}\left[ \norm{x_{k+1}-x_k}^2 \right]\hspace{-0.05cm}.
\end{align}
By taking the expectation and summing for $k$ between $0$ and $N$, we get
\begin{align*}
    &\sum_{k=0}^N \frac{\delta_k}{2} \mathbb{E} \left[\norm{\nabla F(x_k)}^2\right] \\\leq &\sum_{k=0}^N \mathbb{E} \left[F(x_k) -   F(x_{k+1}) \right] + \frac{L_F + \delta_0 \overline{M}^2}{2} \sum_{k=0}^N \mathbb{E}_k\left[ \norm{x_{k+1}-x_k}^2 \right].
\end{align*}
Thanks to Lemma~\ref{lemma:iterates_control_snore_prox_weakly} and the telescopic sum, we get
\begin{align*}
    &\sum_{k=0}^N \frac{\delta_k}{2} \mathbb{E} \left[\norm{\nabla F(x_k)}^2\right] \\\leq &F(x_0) - F^\ast + \frac{L_F + \delta_0 \overline{M}^2}{2} \left( 2\delta_0 (F(x_0) - F^\ast) + \frac{4 \lambda^2 L^2}{\sigma^2 (1 - \delta_0 \rho)} \sum_{k=0}^N \delta_k^2  \right),
\end{align*}
so we  conclude the proof of Lemma~\ref{lemma:intermediate_result} with $A_1 = \bpar{1+\delta_0\bpar{L_F + \delta_0  \max (\rho,M)^2}}$ and $B_1 = \frac{4\lambda^2 L^2}{\sigma^2 }\frac{L_F + \delta_0  \max (\rho,M)^2}{2(1 - \delta_0 \rho)}$.
\end{proof}

We now deduce from Lemma~\ref{lemma:intermediate_result} different convergence results depending on the choice of step-sizes $(\delta_k)_{k \in \N}$.
\subsection{Targeting critical point neighborhood  with constant step-size}\label{ssec:constant_step} In the case of constant step-sizes $\delta_k=\delta$, we first show that the SNORE Prox iterates converge to a critical point neighborhood, parametrized by $\delta$.

\begin{proposition}[Constant step-sizes]\label{prop:constant_step-size_convergence}
Under Assumptions~\ref{ass:den_gradient_step},~\ref{ass:regularities_assumptions} and~\ref{ass:f_smooth}, with $\delta_k = \delta \le \frac{\sigma^2}{\lambda (L+1) + \rho \sigma^2}$ for all $k \in \N$, there exist $A_2, B_2 \in \R_+$, such that the iterates $(x_k)_{k \in \N}$ of the SNORE Prox algorithm verify
\begin{equation}
     \frac{1}{N+1}\sum_{k=0}^N  \eE\left[\norm{\nabla F(x_k)}^2 \right] \leq \frac{A_2}{\delta(N+1)}\left( F(x_0) - F^\ast \right)\nonumber + B_2 \delta.
\end{equation}
\end{proposition}

Proposition~\ref{prop:constant_step-size_convergence} indicates that asymptotically, $\|\nabla F(x_k)\|$ is upper-bounded by the level of noise $\delta$. The inherent stochasticity of SNORE Prox prevents the iterations to converge to a critical point of $F$. Nevertheless, we can observe that the smaller the value of $\delta$, the better the accuracy of  the algorithm.

In order to make $\|\nabla F(x_k)\|$ as close to zero as possible with a constant step-size scheme,  specific mechanisms such as variance reduction~\cite{xiao2014proximal,j2016proximal} or momentum~\cite{gao2024non} are usually added to SPGD. As~\ref{eq:SNORE_Prox} does not include such kind of strategy, it  only converges to a neighborhood of a critical point.

\begin{proof}
Plugging $\delta_k = \delta$, $\forall k \in \N$ in Lemma~\ref{lemma:intermediate_result} and dividing
by $\delta (N+1)$ we get
\begin{align*}
      &\frac{1}{N+1}\sum_{k=0}^N  \eE\left[\norm{\nabla F(x_k)}^2 \right]\\
      \leq &\underbrace{2\bpar{1+\delta_0\bpar{L_F + \delta_0  \overline{M}^2}}}_{:=A_2}\frac1{\delta(N+1)}\left( F(x_0) - F^\ast \right) + \underbrace{\frac{4\lambda^2 L^2}{\sigma^2}\frac{L_F + \delta_0  \overline{M}^2}{(1 - \delta_0 \rho)}}_{:=B_2} \delta.\;\;\;\;\; \qedhere
\end{align*}
\end{proof}

\subsection{Convergence to critical points with non-increasing step-sizes}\label{ssec:varying_step}
We finally derive from Lemma~\ref{lemma:intermediate_result} a convergence result to a critical point of $F$ when considering the SNORE Prox algorithm with non-increasing step-sizes.

\begin{proposition}[non-increasing step-sizes]\label{prop:varying_step-sizes}
Under Assumption~\ref{ass:den_gradient_step},~\ref{ass:regularities_assumptions} and~\ref{ass:f_smooth}, assume that $\delta_0 \le \frac{\sigma^2}{\lambda (L+1) + \rho \sigma^2}$ and that $(\delta_k)_{k \in \N}$ is a non-increasing sequence of step-sizes. Then there exist $A_3, B_3 \in \R_+$ such that the iterates $(x_k)_{k \in \N}$ of the SNORE Prox algorithm verify
\begin{align}\label{prop:non_incre}
    \min_{k=0,\dots,N}\eE\left[\norm{\nabla F(x_k)}^2 \right] &\leq \frac{A_3}{\sum_{k=0}^N \delta_k } \left( F(x_0) - F^\ast \right) +B_3\frac{ \sum_{k=0}^N \delta_k^2}{ \sum_{k=0}^N \delta_k}.
\end{align}
Moreover, if Assumption~\ref{ass:step-size_to_zero} also holds, the iterates $(x_k)_{k \in \N}$ verify almost surely
 \begin{equation}\label{eq:cvg_series}
     \sum_{k \in \N} \delta_k\norm{\nabla F(x_k)}^2 < + \infty.
 \end{equation}
\end{proposition}

\begin{remark}
Classically, the term $\sum_{k=0}^N \delta_k^2 / \sum_{k=0}^N \delta_k$ appearing in Proposition~\ref{prop:varying_step-sizes} indicates a trade-off on the choice of the step-sizes in order to ensure convergence to a critical point, see Assumption~\ref{ass:step-size_tradeoff}. Although such step-size rules are classic in stochastic gradient descent optimization, to our knowledge it is not the case considering SPGD~\cite{ghadimi2016mini}.
Our analysis thus allows for similar step-size trade-off and convergence guarantees within a stochastic proximal scheme.
\end{remark}

\begin{proof}
 We have 
\begin{equation}
    \sum_{k=0}^N  \frac{\delta_k}{2}\eE\left[\norm{\nabla F(x_k)}^2 \right] \geq \min_{k=0,\dots,N}\frac{1}{2}\eE\left[\norm{\nabla F(x_k)}^2 \right]  \sum_{k=0}^N \delta_k.
\end{equation}
Then using Lemma~\ref{lemma:intermediate_result}, we obtain  relation~\eqref{prop:non_incre}:
\begin{align}
    \min_{k=0,\dots,N}\eE\left[\norm{\nabla F(x_k)}^2 \right] \leq &\underbrace{2\bpar{1+\delta_0\bpar{L_F + \delta_0 \overline{M}^2}}}_{:=A_3}\frac{\eE\left[ F(x_0) - F^\ast \right]}{\sum_{k=0}^N \delta_k }\nonumber\\
    &+\underbrace{\frac{4\lambda^2 L^2(L_F + \delta_0 \overline{M}^2)}{\sigma^2(1-\delta_0 \rho) }}_{:=B_3}\frac{ \sum_{k=0}^N \delta_k^2}{ \sum_{k=0}^N \delta_k}.
\end{align}
Next, in order to show~\eqref{eq:cvg_series}, we  subtract $F^\ast$ from each side of Equation~\eqref{eq:bound_gradient} 
 \begin{align}
           \mathbb{E}_k \left[ F(x_{k+1}) - F^\ast \right]\leq& F(x_k) - F^\ast -\frac{\delta_k}{2}\norm{\nabla F(x_k)}^2 \nonumber\\
           &+ \bpar{\frac{L_F + \delta_k \overline{M}^2}{2}}\mathbb{E}_k\left[ \norm{x_{k+1}-x_k}^2 \right].
 \end{align}
 Note that $F(x_k)- F^\ast$, $\frac{\delta_k}{2}\norm{\nabla F(x_k)}^2 $ and $\mathbb{E}_k[ \norm{x_{k+1}-x_k}^2 ] $ are non-negative sequences. Under Assumptions~\ref{ass:den_gradient_step},~\ref{ass:regularities_assumptions} and~\ref{ass:step-size_to_zero}, according to Lemma~\ref{lemma:iterates_control_snore_prox_weakly} we have $\sum_{k \in \N}\mathbb{E}_k[ \norm{x_{k+1}-x_k}^2] < + \infty$. Using the Robbins-Siegmund theorem (Theorem~1 in \cite{RobbinsSiegmund}), we have almost surely that $\sum_{k \in \N}\delta_k\norm{\nabla F(x_k)}^2 < + \infty$.
\end{proof}

\begin{remark}
Note that Proposition~\ref{prop:varying_step-sizes} provides rates of convergence. If $\delta_k = \frac{c}{k^{\alpha}}$, with $\alpha \in (\frac{1}{2}, 1)$, then we get $\min_{k=0,\dots,N}\eE[\norm{\nabla F(x_k)}^2 ] = \mathcal{O}(N^{\alpha - 1})$. Moreover, if $\delta_k = \frac{c}{k}$, then we obtain $\min_{k=0,\dots,N}\eE[\norm{\nabla F(x_k)}^2 ] = \mathcal{O}(\frac{1}{\log(N)})$.
\end{remark}

Proposition~\ref{prop:varying_step-sizes} indicates that with an appropriate choice of step-sizes (see Assumption~\ref{ass:step-size_tradeoff}), one can obtain convergence of the iterates to a critical point of problem~\eqref{eq:ideal_problem}.
\begin{assumption}\label{ass:step-size_tradeoff}
The step-size sequence $(\delta_k)_{k \in \N}$ is non-increasing and verifies
    \begin{equation}
    \sum_{k\in \N}\delta_k = + \infty, \quad \sum_{k\in \N}\delta_k^2 < + \infty.
\end{equation}
\end{assumption}

\begin{proposition}\label{prop:decreasing_to_zero}
Under Assumptions~\ref{ass:den_gradient_step},~\ref{ass:regularities_assumptions},~\ref{ass:f_smooth} and~\ref{ass:step-size_tradeoff}, for $\delta_0 \le \frac{\sigma^2}{\lambda (L+1) + \rho \sigma^2}$, the iterates $(x_k)_{k \in \N}$ of the SNORE Prox algorithm verify almost surely
 \begin{align}
      \min_{k=0,\dots,N}\norm{\nabla F(x_k)} &\underset{N \to +\infty}{\to} 0 \label{eq:lim_min_zero}\\
      \underset{k \to +\infty}{\lim \inf}~\|\nabla F(x_k)\| &= 0. \label{eq:lim_inf_zero}
\end{align}
\end{proposition}

Proposition~\ref{prop:decreasing_to_zero} states that with an appropriate choice of step-sizes, zero is an accumulation point of $\|\nabla F(x_k)\|$.

\begin{proof}
\textbf{Proof of equation~\eqref{eq:lim_min_zero}} We first define the random variable $Y_N = \min_{k=0,\dots,N}\norm{\nabla F(x_k)}$. According to Proposition~\ref{prop:varying_step-sizes}, under Assumption~\ref{ass:step-size_tradeoff} we have $\min_{k=0,\dots,N}\hspace{-1pt}\eE[\norm{\nabla F(x_k)}^2 ] \hspace{-5pt}\underset{N \to +\infty}{\to} \hspace{-5pt}0$.
With the Cauchy-Schwarz inequality we get
\begin{align*}
&\min_{k=0,\dots,N}\eE\left[\norm{\nabla F(x_k)}^2 \right] \ge \min_{k=0,\dots,N}\eE\left[\norm{\nabla F(x_k)} \right]^2 = \left(\min_{k=0,\dots,N}\eE\left[\norm{\nabla F(x_k)} \right] \right)^2 \\
    &\ge \left(\eE\left[\min_{k=0,\dots,N} \norm{\nabla F(x_k)} \right] \right)^2 = \left(\eE\left[Y_N\right] \right)^2 \underset{N \to +\infty}{\to} 0.
\end{align*}
Thus we get $\eE\left[Y_N\right] \to 0.$
Moreover, by definition of $Y_N$, it is a non-increasing and non-negative sequence of random variables. Therefore, if we denote $\Omega$ the space of realizations, $\forall \omega \in \Omega$, $(Y_N(\omega))_{N \in \N}$ is a non-negative, non-increasing sequence, so it is converging to a limit $l(\omega) \in \R^+$.

For $\epsilon > 0$, If $l(\omega) \ge \epsilon$, then $\forall N \in \N, Y_N(\omega) \ge \epsilon$. So $\eE(Y_N) \ge \eE(\mathbbm{1}_{l \ge \epsilon} Y_N) \ge \epsilon \mathbb{P}(l \ge \epsilon)$. Due to  $\eE(Y_N) \to 0$, we get $\forall \epsilon > 0$, $\mathbb{P}(l \ge \epsilon) = 0$. Then we get $\mathbb{P}(l > 0) = \mathbb{P}\left(\lim_{\epsilon \to 0} (l \ge \epsilon)\right) = \lim_{\epsilon \to 0}\mathbb{P}(l \ge \epsilon) = 0$.

Finally, almost surely $l = 0$, i.e. $\min_{k=0,\dots,N}\norm{\nabla F(x_k)} \to 0$.

\noindent
\textbf{Proof of equation~\eqref{eq:lim_inf_zero}} We make a proof by contradiction and assume that a.s. $\underset{k \to +\infty}{\lim \inf}~\norm{\nabla F(x_k)}^2 = c > 0$. It means that a.s. there exists $N_0 \in \N$ such that $\forall k \ge N_0$, $\norm{\nabla F(x_k)}^2 \ge c$. We thus get
\begin{equation}
     \sum_{k \in \N}\delta_k\norm{\nabla F(x_k)}^2 \ge  \sum_{k \ge N_0}\delta_k\norm{\nabla F(x_k)}^2 \ge c\sum_{k \ge N_0}\delta_k = +\infty,
\end{equation}
where the last equality holds thanks to Assumption~\ref{ass:step-size_tradeoff}.
However, according to Proposition~\ref{prop:varying_step-sizes}, we have
 \begin{equation}
     \sum_{k \in \N} \delta_k\norm{\nabla F(x_k)}^2 < + \infty,~\text{a.s.}
 \end{equation}
which leads to a contradiction. So we have,  almost surely,  $\underset{k \to +\infty}{\lim \inf}~\norm{\nabla F(x_k)}^2 = 0$.
By continuity of the squared function, we get the desired result.
\end{proof}

\section{Experiments}

\begin{figure*}[!ht]
    \centering
    \includegraphics[width=\textwidth]{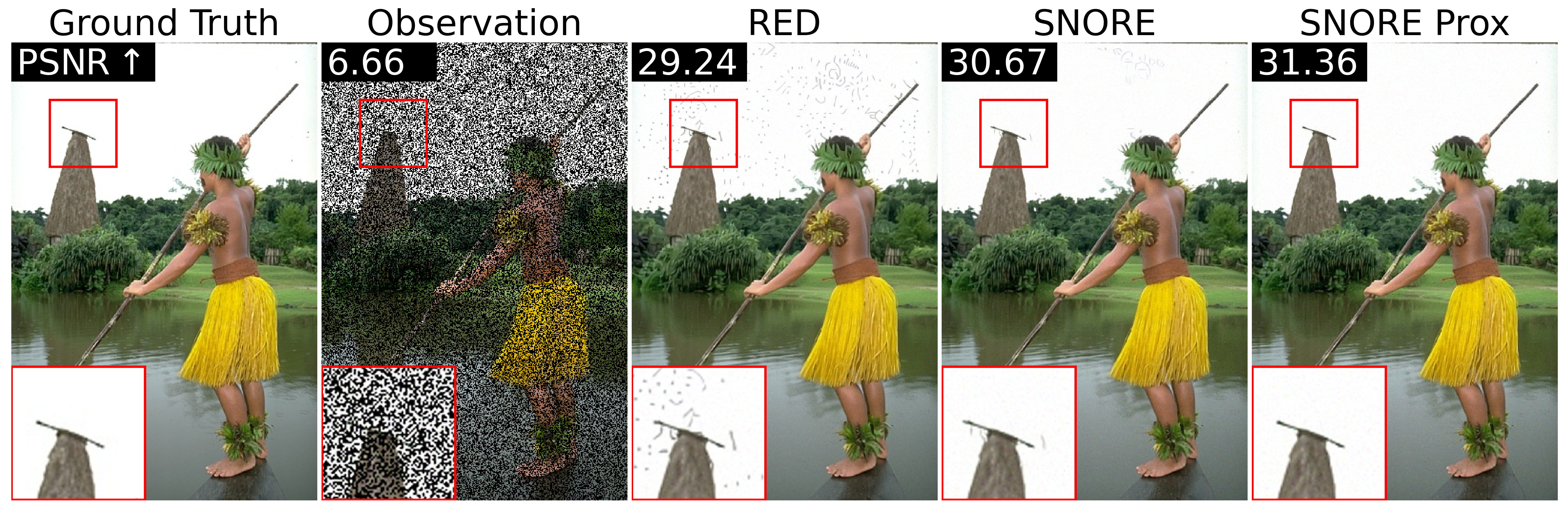}
    \caption{Inpainting, $50\%$ missing pixels and a noise level $\sigma_y = 5/255$, with various restoration methods with a GS-denoiser trained on natural images. Note that SNORE Prox produces better qualitative results than SNORE or RED. \vspace*{-0.4cm}
    }
    \label{fig:restauration}
\end{figure*}

In this section, we evaluate quantitatively and qualitatively the practical gain of SNORE Prox. We focus on inpainting with random missing pixels. The denoiser is the gradient-step DRUNet with pre-trained weights provided in~\cite{hurault2022gradient}. Experiments are run {\bf with constant step-sizes}. More details about parameter setting and additional experiments are provided in Appendix~\ref{sec:more_expe}.
On Figure~\ref{fig:restauration}, we present qualitative result of RED, SNORE and SNORE Prox. We observe that SNORE Prox succeeds to reduce artifacts compared to RED or SNORE.

\begin{figure}[!ht]
    \centering
    \begin{minipage}[t]{0.45\textwidth}\vspace{-4.3cm}
        \centering
        \resizebox{\linewidth}{!}{
        \begin{tabular}{| c | c | c c c |}
        \hline
        & Method & PSNR$\uparrow$ & SSIM$\uparrow$ & LPIPS$\downarrow$ \\
        \hline
        \multirow{7}{*}{\centering $\sigma_y = 0$} & RED & 31.26 & 0.91 & 0.07  \\
        & RED Prox & 30.31 & 0.89 & 0.12 \\ 
        & SNORE & 31.32 & 0.91 & \underline{0.05}  \\
        & SNORE Prox & \underline{31.69} & \underline{0.92} & \textbf{0.04}\\
        & Ann-SNORE & 31.65 & \underline{0.92} & \textbf{0.04}  \\
        & Ann-SNORE Prox & \textbf{31.94} & \textbf{0.93} & \textbf{0.04} \\
        & DiffPIR & 29.57 & 0.87 & 0.07 \\
        \hline
        \multirow{7}{*}{\centering $\sigma_y = \frac{5}{255}$} & RED & 30.18 & \underline{0.87} & 0.06 \\
        & RED Prox & 30.21 & \textbf{0.88} & 0.12 \\
        & SNORE & 30.15 & 0.86 & \underline{0.05} \\
        & SNORE Prox & 30.21 & 0.86 & \underline{0.05} \\
        & Ann-SNORE & \underline{30.34} & \underline{0.87} & \textbf{0.04} \\
        & Ann-SNORE Prox & \textbf{30.44} & \textbf{0.88} & \textbf{0.04} \\
        & DiffPIR & 29.98 & \textbf{0.88} & 0.06 \\
        \hline
        \end{tabular}
        }
        \caption{Inpainting result for random missing pixel with probability $p = 0.5$ on CBSD68 dataset. Best and second-best results are respectively displayed in bold and underlined. Algorithms are run with constant step-sizes.}
        \label{table:inpainting}
    \end{minipage}
    \hfill
    \begin{minipage}[t]{0.52\textwidth}
        \centering
        \includegraphics[width=\linewidth]{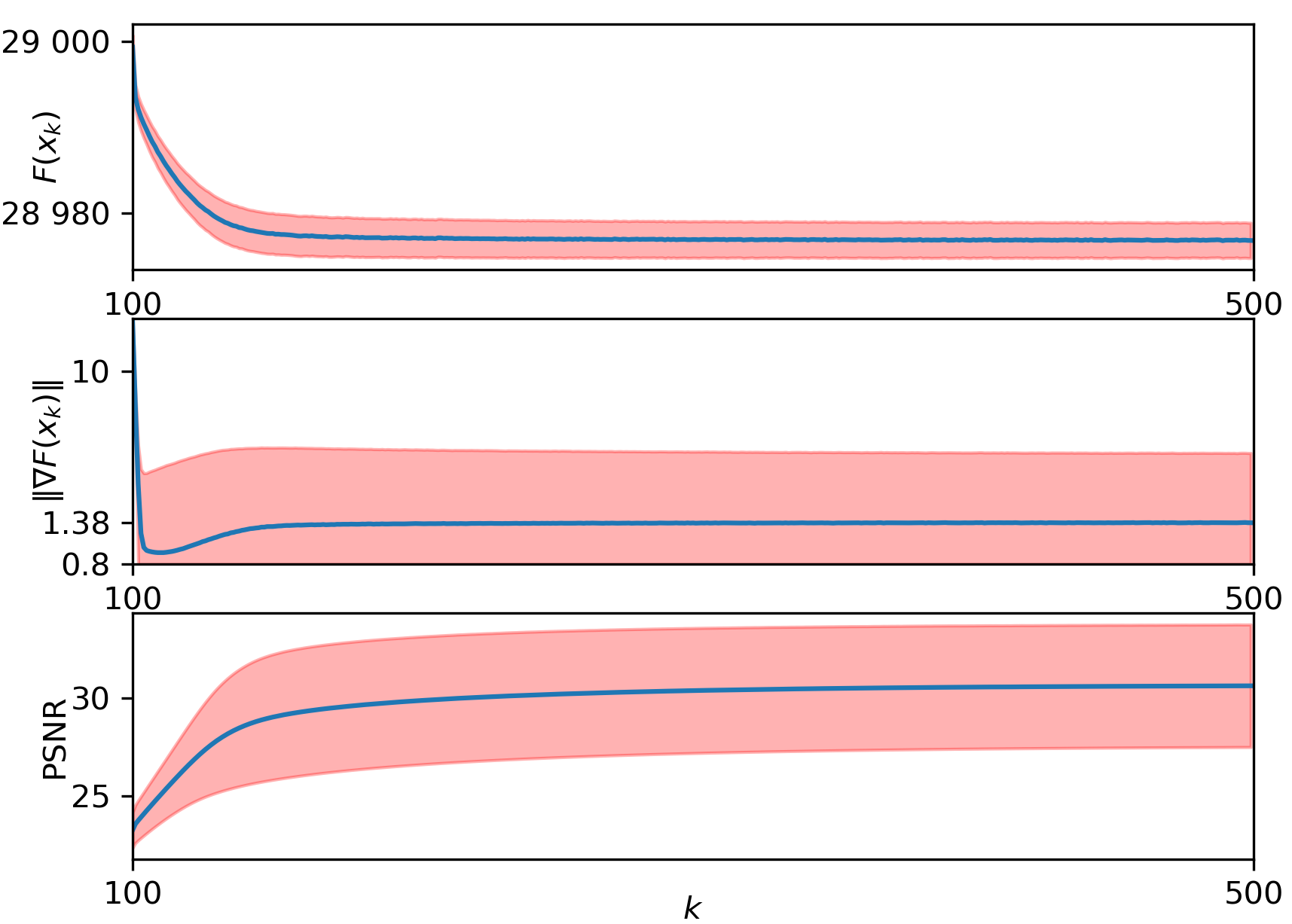}
        \caption{Convergence results of SNORE Prox with constant step-sizes in terms of decrease of the objective function $F(x_k)$, of gradient of the objective function $\nabla F(x_k)$ and PNSR. The target problem is inpainting with $50\%$ missing pixels and a noise level of $\sigma_y = 5/255$. Results are averaged on $10$ images extracted from the CBSD68 dataset.}
        \label{fig:convergence}
    \end{minipage}
\end{figure}

On Figure~\ref{table:inpainting}, we present quantitative results of SNORE Prox compared with various restoration methods on inpainting without noise in the observation or with a noise level $\sigma_y = 5/255$. Some results without noise are extracted from ~\cite[Table 2]{renaud2024plugandplayimagerestorationstochastic}. Ann-SNORE and Ann-SNORE Prox are annealed version of SNORE and SNORE Prox, where $\sigma$ and $\lambda$ vary during iterations, as proposed in~\cite{renaud2024plugandplayimagerestorationstochastic}. Note that proximal versions of SNORE improves the accuracy compared to SNORE.

On Figure~\ref{fig:convergence}, we observe the convergence of SNORE Prox algorithm. We see that the objective function is effectively minimized. Note that the gradient of $\nabla F(x_k)$ does not go to zero as this experiment is realized under the constant step-size regime.

\section{Conclusion}
In this paper, we provide a new convergence analysis for the SNORE Prox algorithm. We generalize an existing result on SPGD with a weakly-convex proximal function (Lemma~\ref{lemma:iterates_control_snore_prox_weakly}). We  present a new convergence result under the $M$-smoothness of $f$ in Lemma~\ref{lemma:intermediate_result}. From this technical analysis, we deduce the convergence of SNORE Prox to a neighborhood of a critical point of the target functional in the constant step-size regime (Proposition~\ref{prop:constant_step-size_convergence}); and to a critical point in the non-increasing step-size regime (Proposition~\ref{prop:varying_step-sizes}-\ref{prop:decreasing_to_zero}). Finally, we provide numerical experiments that illustrate  the practical convergence of SNORE Prox as well as its performance for image inpainting.

\section{Acknowledgment}
This study has been carried out with financial support from the French Direction G\'en\'erale de l'Armement and ANR project PEPR PDE-AI. Experiments presented in this paper were carried out using the PlaFRIM experimental testbed, supported by Inria, CNRS (LABRI and IMB), Université de Bordeaux, Bordeaux INP and Conseil Régional d’Aquitaine (see https://www.plafrim.fr). We thank Arthur Leclaire for his time and discussions.

\bibliography{ref}
\bibliographystyle{abbrv}

\newpage
\appendix

\section{Proof of Lemma~\ref{lemma:iterates_control_snore_prox_weakly}}\label{sec:proof_lemma_weakly}

In this section, we provide the proof of Lemma~\ref{lemma:iterates_control_snore_prox_weakly}. For completeness, we first demonstrate two classical lemmas on weakly convex functions.

\begin{lemma}\label{lemma:inequality_weakly_cvx}
If $f$ is $\rho$-weakly convex and differentiable, then we have $\forall x, y \in \R^d$,
\begin{align}
    \langle \nabla f(x), y - x\rangle \le f(y) - f(x) + \frac{\rho}{2} \|y-x\|^2.
\end{align}
\end{lemma}
\begin{proof}
$f + \frac{\rho}{2} \|\cdot\|^2$ is convex. Thus, $\forall x, y \in \R^d$, we have
\begin{align*}
    \langle \nabla f(x) + \rho x, y - x  \rangle &\le f(y) + \frac{\rho}{2} \|y\|^2 - f(x) - \frac{\rho}{2} \|x\|^2 \\
    \langle \nabla f(x), y - x  \rangle &\le f(y)- f(x) + \frac{\rho}{2} \|y\|^2  + \frac{\rho}{2} \|x\|^2 - \rho \langle x, y \rangle \\
    \langle \nabla f(x), y - x  \rangle &\le f(y)- f(x) + \frac{\rho}{2} \|y - x\|^2.&\qedhere
\end{align*}
\end{proof}

\begin{lemma}\label{lemma:control_prox_map_weakly}
For $f$ $\rho$-weakly convex with $\rho \delta_k < 1$ and differentiable, we have $\forall x, y \in \R^d$,
\begin{align}
    \|\text{Prox}_{\delta_k f}(x) - \text{Prox}_{\delta_k f}(y)\| \le \frac{\lambda}{1 - \delta_k \rho} \|x - y\|.
\end{align}
\end{lemma}
\begin{proof}
If we note $u = \text{Prox}_{\delta_k f}(x)$ and $v =\text{Prox}_{\delta_k f}(y)$.
By the optimal condition of the proximal operator, we get
\begin{align*}
    \frac{1}{\delta_k} (u - x) + \nabla f(u) = 0 \\
    \frac{1}{\delta_k} (v - y) + \nabla f(v) = 0.
\end{align*}
Moreover by Lemma~\ref{lemma:inequality_weakly_cvx}, we have
\begin{align*}
    \langle \nabla f(u), v - u \rangle &\le f(v) - f(u) +\frac{\rho}{2} \|u-v\|^2 \\
    \langle \nabla f(v), u - v \rangle &\le f(u) - f(v) +\frac{\rho}{2} \|u-v\|^2
\end{align*}
So we get
\begin{align*}
    \frac{1}{\delta_k}\langle u - x, v - u \rangle &\ge  f(u) - f(v) - \frac{\rho}{2} \|u-v\|^2\\
    \frac{1}{\delta_k}\langle v - y, u - v \rangle &\ge  f(v) - f(u) - \frac{\rho}{2} \|u-v\|^2.
\end{align*}
By summing the two previous ineaqualities, we get
\begin{align*}
    \frac{1}{\delta_k} \langle u - v + y - x, v - u \rangle \ge - \rho \|u-v\|^2 \\
    \|x-y\| \|u-v\| \ge \langle y - x, v - u \rangle \ge \left(1 - \delta_k \rho\right) \|u-v\|^2.
\end{align*}
So we get the desire result
\begin{align*}
    \|u-v\| \le \frac{1}{1 - \delta_k \rho} \|x-y\|.
\end{align*}
\end{proof}

We recall here the statement of Lemma~\ref{lemma:iterates_control_snore_prox_weakly}.
\begin{lemma}
Under Assumptions~\ref{ass:den_gradient_step} and~\ref{ass:regularities_assumptions}, for $\delta_0 \le \frac{\sigma^2}{\lambda (L+1) + \rho \sigma^2}$ and $(\delta_k)_{k \in \N}$ a non-increasing sequence of step-size, we get
\begin{align}
    \sum_{k=0}^{N-1} \eE_k\left(\|x_{k+1} - x_k\|^2\right) \le 2\delta_0 (F(x_0) - F^\ast) + \frac{4 \lambda^2 L^2}{\sigma^2 (1 - \delta_0 \rho)} \sum_{k=0}^N \delta_k^2.
\end{align}
\end{lemma}

This proof is a generalization of Theorem 2 in~\cite{ghadimi2016mini} with a weakly convex $f$. We follow the scheme of their proof.

\begin{proof}
First, we recall that for a $\rho$-weakly convex function $f$, for all $\delta<\frac{1}{\rho}$, problem $\argmin_{z \in \R^d} \frac{1}{2\delta}\|x-z\|^2+f(z)$ is strongly convex so that $\text{Prox}_{\delta f}$ is univalued. Next we introduce the quantity $G_{k}$ from the  proximal mapping~\eqref{snore_prox}~as 
\begin{align}
    G_{k} = \frac{x_k - x_{k+1}}{\delta_k}= \frac{1}{\delta_k}\left(x_k - \text{Prox}_{\delta_k f}\left( x_k - \delta_k \lambda \tilde \nabla g_{\sigma}(x_k) \right)\right).\label{def_G2}
\end{align}
From Lemma~\ref{prop:gsmooth}, we have that $\nabla g_\sigma$ is $L_{h, \sigma} = (1+L)/\sigma^2$-Lipschitz, which gives
\begin{align}
    g_{\sigma}(x_{k+1}) &\le g_{\sigma}(x_{k}) + \langle \nabla g_{\sigma}(x_{k}), x_{k+1} - x_k \rangle +\frac{L_{h, \sigma}}{2} \|x_{k+1} - x_k \|^2 \\
    &= g_{\sigma}(x_{k}) - \delta_k \langle \nabla g_{\sigma}(x_{k}), G_{k} \rangle +\frac{L_{h, \sigma} \delta_k^2}{2} \|G_{k}\|^2 \\
    &= g_{\sigma}(x_{k}) - \delta_k \langle \tilde \nabla g_{\sigma}(x_{k}), G_{k} \rangle +\frac{L_{h, \sigma} \delta_k^2}{2} \|G_{k}\|^2 + \delta_k \langle \zeta_k, G_{k} \rangle \label{eq:first_inequalities2}
\end{align}
The optimal condition of the proximal operator implies that 
\begin{align}
    \frac{x_{k+1}-x_k}{\delta_k} + \lambda \tilde \nabla g_{\sigma}(x_{k}) + \nabla f(x_{k+1}) = 0.
\end{align}
So, $G_{k}$ can also be expressed as
\begin{align}\label{eq:prox_map_formula2}
    G_{k} = \lambda \tilde \nabla g_{\sigma}(x_{k}) + \nabla f(x_{k+1}).
\end{align}
By using equations~\eqref{eq:first_inequalities2} and~\eqref{eq:prox_map_formula2}, we get
\begin{align}
    &g_{\sigma}(x_{k+1}) \le g_{\sigma}(x_{k}) - \delta_k \langle \frac{1}{\lambda}\left( G_{k} - \nabla f(x_{k+1}) \right), G_{k} \rangle +\frac{L_{h, \sigma} \delta_k^2}{2} \|G_{k}\|^2 + \delta_k \langle \zeta_k, G_{k} \rangle \\
    &= g_{\sigma}(x_{k}) + \left(\frac{L_{h, \sigma} \delta_k^2}{2} - \frac{\delta_k}{\lambda} \right) \|G_{k}\|^2 + \frac{\delta_k}{\lambda} \langle \nabla f(x_{k+1}) , G_{k} \rangle + \delta_k \langle \zeta_k, G_{k} \rangle \\
    &= g_{\sigma}(x_{k}) + \left(\frac{L_{h, \sigma} \delta_k^2}{2} - \frac{\delta_k}{\lambda} \right) \|G_{k}\|^2 + \frac{1}{\lambda} \langle \nabla f(x_{k+1}) , x_k - x_{k+1} \rangle + \delta_k \langle \zeta_k, G_{k} \rangle. \label{eq:inequality_before_cvx2}
\end{align}
Then using equation~\eqref{eq:inequality_before_cvx2} and Lemma~\ref{lemma:inequality_weakly_cvx}, we get
\begin{align*}
   & g_{\sigma}(x_{k+1}) \\\le\, &g_{\sigma}(x_{k}) + \left(\frac{L_{h, \sigma} \delta_k^2}{2} - \frac{\delta_k}{\lambda} \right) \|G_{k}\|^2 + \frac{1}{\lambda} \left( f(x_k) - f(x_{k+1}) + \frac{\rho}{2} \|x_{k+1}-x_k \|^2 \right)\\& + \delta_k \langle \zeta_k, G_{k} \rangle \\
    =\, &g_{\sigma}(x_{k}) + \left(\frac{L_{h, \sigma} \delta_k^2}{2} - \frac{\delta_k}{\lambda} \right) \|G_{k}\|^2 + \frac{1}{\lambda} \left( f(x_k) - f(x_{k+1}) + \frac{\rho \delta_k^2}{2} \|G_{k}\|^2 \right) \\&+ \delta_k \langle \zeta_k, G_{k} \rangle \\
    =\,& g_{\sigma}(x_{k}) + \left(\frac{L_{h, \sigma} \delta_k^2}{2} + \frac{\rho \delta_k^2}{2\lambda} - \frac{\delta_k}{\lambda} \right) \|G_{k}\|^2 + \frac{1}{\lambda} \left( f(x_k) - f(x_{k+1}) \right) + \delta_k \langle \zeta_k, G_{k} \rangle
\end{align*}
Next we introduce $\overline{G}_{k}$ as
\begin{align}
    \overline{G}_{k} = \frac{1}{\delta_k}\left(x_k - \text{Prox}_{\delta_k f}\left( x_k - \delta_k \lambda \nabla g_{\sigma}(x_k) \right)\right).\label{def_Gbar}
\end{align}
By re-arranging terms we obtain
\begin{align}
   &\delta_k  \left(1 - \frac{(\lambda L_{h, \sigma} + \rho) \delta_k}{2} \right) \|G_{k}\|^2\nonumber \\\le& F(x_{k}) - F(x_{k+1}) + \delta_k \lambda \langle \zeta_k, \overline{G}_{k} \rangle + \delta_k \lambda \langle \zeta_k, G_{k} - \overline{G}_{k} \rangle.
\end{align}
Taking the expectation $\eE_k$ with respect with $x_k$, and using that $\delta_k \le \delta_0 \le \frac{1}{\lambda L_{h, \sigma} + \rho}$ and $\eE_k(\zeta_k) = 0$, we get
\begin{align}
   \delta_k \eE_k\left(\|G_{k}\|^2\right) \le 2 \eE_k\left(F(x_{k}) - F(x_{k+1})\right) + 2 \delta_k \lambda \eE_k\left(\langle \zeta_k, G_{k} - \overline{G}_{k} \rangle\right).\label{eq:control_with_scalar_product}
\end{align}
By using Lemma~\ref{lemma:control_prox_map_weakly} on equation~\eqref{eq:control_with_scalar_product}, we get
\begin{align}
   &\delta_k \eE_k\left(\|G_{k}\|^2\right) \le 2 \eE_k\left(F(x_{k}) - F(x_{k+1})\right) +  \frac{2 \delta_k\lambda^2}{1 - \delta_k \rho} \eE_k\left(\| \zeta_k \|^2\right).
\end{align}
Moreover, we have
\begin{align}
    \eE_k(\|\zeta_k\|^2) &= \frac{1}{\sigma^4} \eE_k(\|x_k - D_{\sigma}(x_k + \sigma z_{k+1}) - \eE_{z \sim \mathcal{N}(0, I_d)}\left(x_k - D_{\sigma}(x_k + \sigma z)\right)\|^2)\nonumber  \\
    &\le \frac{1}{\sigma^4} \eE_k(\eE_{z \sim \mathcal{N}(0, I_d)}\left( \|D_{\sigma}(x_k + \sigma z_{k+1}) - D_{\sigma}(x_k + \sigma z)\|\right)^2) \nonumber \\
    &\le \frac{L^2}{\sigma^2} \eE_k(\eE_{z \sim \mathcal{N}(0, I_d)}\left( \|z_{k+1} - z\|\right)^2) \le \frac{L^2}{\sigma^2} \eE_k(1 + \|z_{k+1}\|^2)\nonumber  \\
    &\le \frac{2 L^2}{\sigma^2}.
\end{align}
Thus, we obtain
\begin{align}\label{eq:bound_gadient_mapping}
    \delta_k \eE_k\left(\|G_{k}\|^2\right) \le 2 \eE_k\left(F(x_{k}) - F(x_{k+1})\right) +  \frac{4 \delta_k\lambda^2 L^2}{\sigma^2 (1 - \delta_0 \rho)},
\end{align}
where we use that $\delta_k \le \delta_0$.
Then, as $G_k = \frac{x_k - x_{k+1}}{\delta_k}$, we have
\begin{align}
    \eE_k\left(\|x_{k+1} - x_k\|^2\right) \le 2 \delta_k \eE_k\left(F(x_{k}) - F(x_{k+1})\right) +  \frac{4 \delta_k^2 \lambda^2 L^2}{\sigma^2 (1 - \delta_0 \rho)}.
\end{align}

By summing for k between $0$ and $N-1$, thanks to $(\delta_k)_{k \in \N}$ being non-increasing, we get
\begin{align*}
\sum_{k=0}^{N-1} \eE_k\left(\|x_{k+1} - x_k\|^2\right) &\le 
2\sum_{k=0}^N \delta_k (F(x_k) - F(x_{k+1})) + \frac{4\lambda^2 L^2}{\sigma^2 (1 - \delta_0 \rho)} \sum_{k=0}^N \delta_k^2\nonumber\\
&\le 2 \delta_0 (F(x_0) - F^\ast) + \frac{4\lambda^2 L^2}{\sigma^2 (1 - \delta_0 \rho)} \sum_{k=0}^N \delta_k^2.&\qedhere
\end{align*}
\end{proof}

\section{Additional numerical details}\label{sec:more_expe}
\paragraph{Proximal operator for Image Inpainting}The forward model is $y = A x + n$, with $A$ a diagonal matrix of $0$ and $1$ and $n \sim \mathcal{N}(0, \sigma_y^2 I_m)$. 

If $\sigma_y = 0$, then $f(x) = i_{A^{-1}(y)}$, with $A^{-1}(y) = \{x \in \R^d | Ax = y\}$ and 
$$i_S(x) = \left\{
\begin{array}{ll}
0&\textrm{if }x \in S \\
+\infty &\textrm{otherwise}.
\end{array}
\right.$$
The proximal step thus reads 
$$\text{Prox}_{\delta f} (x)= )Ay - Ax + x,$$
which corresponds to the orthogonal projection on the convex set $A^{-1}(y)$. In this case, $f$ is not differentiable everywhere, so Assumption~\ref{ass:regularities_assumptions}(b) is not verified.

If $\sigma_y > 0$, then $f(x) = \frac{1}{\sigma_y^2} \|y - A x\|^2$ and we get
$$\text{Prox}_{\delta f}(x) = x +\left(1 + \frac{\sigma_y^2}{\tau}\right)^{-1}(Ay - Ax).$$ 
Moreover, $f$ is convex, so Assumption~\ref{ass:regularities_assumptions} is verified with $\rho = 0$.

\paragraph{Hyper-parameters setting}
For hyper-parameters choice, a grid search has been realized for each method to find the best parameters in term of PSNR. In the case of $\sigma_y = 0$, quantitative results for RED, RED Prox, Ann-SNORE, Ann-SNORE Prox and DiffPIR have been taken from~\cite[Table 2]{renaud2024plugandplayimagerestorationstochastic}. This corresponds to take the parameters suggested by the respective authors.

On Table~\ref{table:parameter_inpainting}, we provide the hyper-parameters used for various methods:
$\delta$ the constant step-size, $\lambda$ the regularization parameter and $\sigma$ the denoiser parameter. For annealed methods, $\lambda$ is increasing between $\lambda_0$ and $\lambda_{m-1}$ along iterations and $\sigma$ is decreasing between $\sigma_0$ and $\sigma_{m-1}$, as detailed in~\cite{renaud2024plugandplayimagerestorationstochastic}. $\lambda$ is increasing when $\sigma$ is decreasing in order to ensure that the regularization maintains a significant weight. For DiffPIR~\cite{zhu2023denoising} with a noise level $\sigma_y = 5/255$, we set $T = 1000$, $t_{start} = 200$, $\zeta = 0.99$ and $\lambda = 0.005$.

\begin{table}[h]
\centering
\resizebox{\linewidth}{!}{
\begin{tabular}{| c |c | c | c| c| c|c|c|c |c|}
\hline
 & Method &
 $\delta$ & $\lambda$ & $\sigma$ &  $n_{init}$ & $\lambda_0$ & $\lambda_{m-1}$ & $\sigma_0$ & $\sigma_{m-1}$ \\
 \hline
 \multirow{7}{*}{\centering $\sigma_y = 0$} & RED & 
 1/$\alpha$ & 0.15 & 10/255 & 10 &&&& \\
 & RED Prox & 
 0.5 & 0.15 & 10/255 & 100 &&&& \\ 
 & SNORE & 
 0.5 & 0.15 & 11/255 & 100 &&&& \\
  &SNORE Prox & 
  2.0 & 0.06 & 7/255 & 100 &&&&\\
  &Ann-SNORE & 
  0.5 &  & & & 0.15 & 0.4& 50/255 & 5/255 \\
  &Ann-SNORE Prox &  
  1.0 &  & & & 0.15 & 0.15& 50/255 & 5/255\\
  \hline
 \multirow{7}{*}{\centering $\sigma_y = \frac{5}{255}$} & RED & 
 0.5 & 0.19 & 10/255 & 100 &&&& \\
 & RED Prox & 
 0.5 & 0.19 & 9/255 & 100 &&&& \\ 
 & SNORE & 
 0.5 & 0.19 & 10/255 & 100 &&&& \\
  &SNORE Prox &  
  2.0 & 0.05 & 8/255 & 100 &&&&\\
  &Ann-SNORE & 
  0.5 &  & & & 0.15 & 0.35& 50/255 & 5/255 \\
  &Ann-SNORE Prox &  
  1.0 &  & & & 0.15 & 0.5& 50/255 & 5/255\\
\hline
\end{tabular}
}\vspace*{.2cm}
\caption{Hyper-parameters setting for image inpainting.}
\label{table:parameter_inpainting}
\end{table}

\paragraph{On the $\nabla F(x_k)$ limit}
On Table~\ref{table:lim_nabla_F}, we observe that the mean value of $\nabla F(x_N)$ for $N = 500$ effectively decreases with the stepsize $\delta$. In practice, we prefer to keep high step-sizes ($\delta = 2.0$), in order to ensure fast convergence.

\begin{table}[h]
\centering
\begin{tabular}{| c |c | c | c | c| c|}
\hline
$\delta$ & 0.05 & 0.1 &0.5 & 1.0 & 2.0 \\
\hline
$\| \nabla F(x_{500}) \|$ & 0.83 & 0.89 & 1.15 & 1.26 & 1.37 \\
\hline
\end{tabular}\vspace*{.2cm}
\caption{Mean values of $\nabla F(x_{500})$ for inpainting with $50\%$ missing pixels and a noise level of $\sigma_y = 5/255$ on $10$ images extracted for CBSD68 dataset. As predicted by Proposition~\ref{prop:constant_step-size_convergence}, we observe that the smaller  $\delta$, the smaller  $\|\nabla F(x_{500})\|$.}
\label{table:lim_nabla_F}
\end{table}

\begin{figure*}[!ht]
    \centering
    \includegraphics[width=\textwidth]{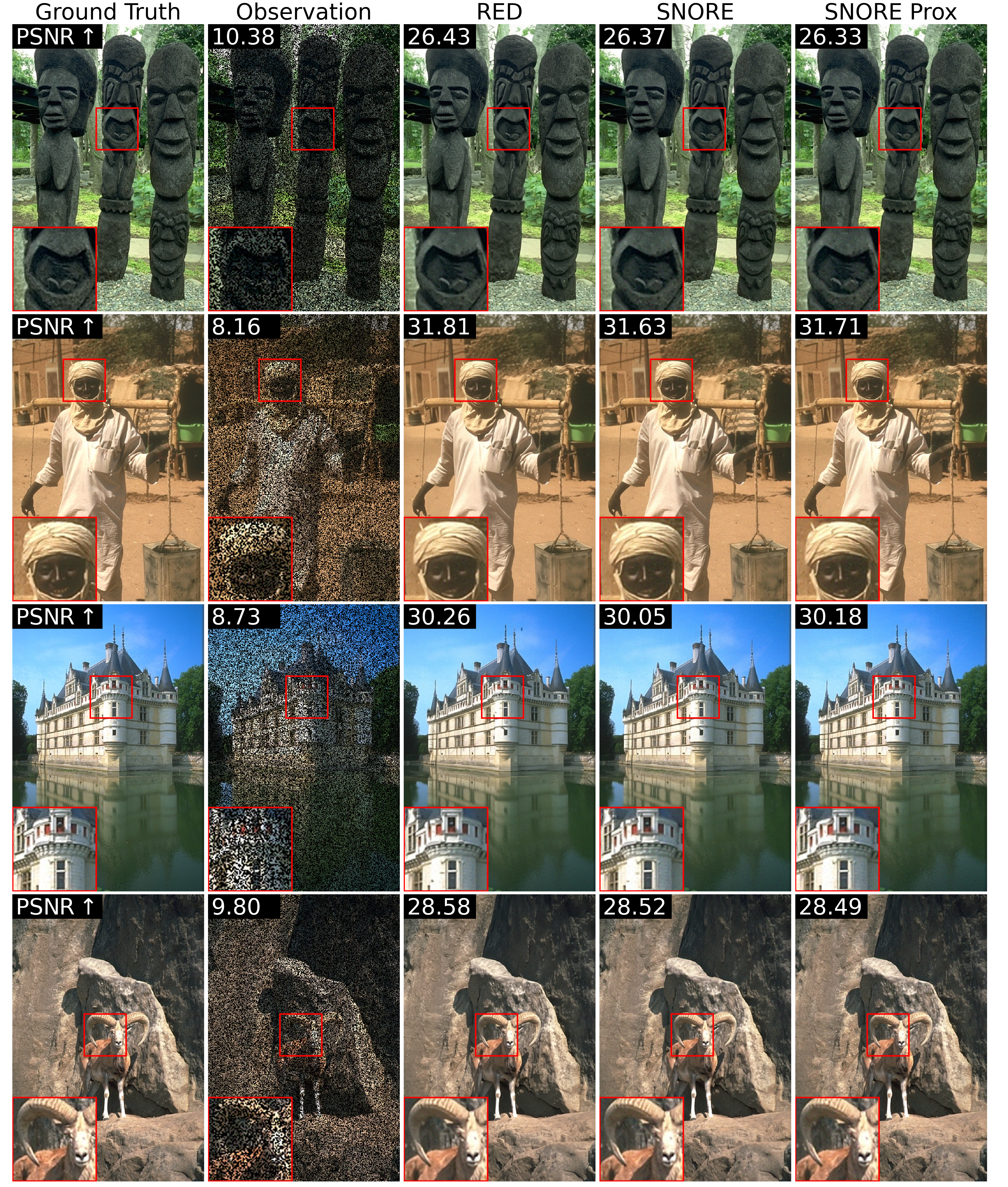}
    \caption{Additional inpainting results with $50\%$ missing pixels and a noise level $\sigma_y = 5/255$, with various restoration methods with a GS-denoiser trained on natural images.\vspace*{-0.4cm}
    }
\end{figure*}

\end{document}